\documentclass[journal]{IEEEtran}
\IEEEoverridecommandlockouts
\usepackage{graphicx} 
\usepackage{amsthm}
\usepackage{amsmath}
\usepackage{amsfonts}
\usepackage{amssymb}
\usepackage{xspace}
\usepackage{bbm}
\usepackage{tcolorbox}
\usepackage{booktabs,caption}
\usepackage[flushleft]{threeparttable}
\usepackage{boxedminipage}
\tcbuselibrary{most}
\newtheorem{definition}{Definition}

\newtheorem{theorem}{Theorem}

\usepackage{tikz}
\usepackage{pgfplots}
\usepackage{makecell}
\usepackage{floatrow}
\usepackage{pifont}
\usepackage[colorlinks=true,urlcolor=black]{hyperref}
\usepackage{algorithm, algorithmicx, algpseudocode}
\newfloatcommand{capbtabbox}{table}[][\FBwidth]
\hypersetup{
colorlinks=true,
linkcolor=blue,
filecolor=magenta,
urlcolor=cyan,
citecolor=green,
}
\pgfplotsset{compat=1.18} 

\definecolor{beni}          {HTML} {FF6666}
\definecolor{mizu}          {HTML} {40C7F4}
\definecolor{kon}           {HTML} {009F85}
\definecolor{ao}            {HTML} {5ED9C4}
\definecolor{anzu}          {HTML} {FFB366}
\definecolor{usumurazaki}   {HTML} {B39CD0}
\definecolor{murazaki}      {HTML} {845EC2}
\definecolor{ruri}          {HTML} {2C73D2}
\definecolor{enji}          {HTML} {D65DB1}
\definecolor{momo}          {HTML} {FF6F91}
\definecolor{kogane}        {HTML} {FFC75F}
\definecolor{ki}            {HTML} {F9F871}

\newtcolorbox{mybox}[2][]
  {colback=white,coltitle=black,colbacktitle = white, enhanced, coltitle=black,boxrule=0.7pt,
    attach boxed title to top left={xshift=4mm,yshift=-2.8mm}, left=0.1 mm,right= 0.1 mm,title=#2,#1}

\usepackage{enumitem}
\usepackage{enumerate}

\usepackage [
    n,
    advantage,
    operators,
    sets,
    adversary,
    landau,
    probability,
    notions,
    logic,
    ff, 
    mm,
    primitives,
    events,
    complexity,
    asymptotics,
    keys,
]{cryptocode}
\createpseudocodeblock{pcb}{center, boxed}{}{}{}
\createprocedureblock{procb}{center, boxed}{}{}{}

\newcommand{\spar}[1]{%
  \medskip%
  \noindent%
  \textbf{#1.}\xspace%
}
\newcommand{\cirn}[1]{\raisebox{-0.7pt}{\ding{#1}}}

\newcommand{\Sysname}{IRAC\xspace}
\newcommand{\Ctx}{\ensuremath{\mathsf{ctx}}\xspace}
\newcommand{\MyPP}{\ensuremath{\mathsf{pp}}\xspace}

\newcommand{\MySetup}{\ensuremath{\mathsf{Setup}}\xspace}
\newcommand{\IssuerSetup}{\ensuremath{\mathsf{IssuerSetup}}\xspace}
\newcommand{\AddAttributes}{\ensuremath{\mathsf{AddAttributes}}\xspace}
\newcommand{\Revoke}{\ensuremath{\mathsf{Revoke}}\xspace}
\newcommand{\IssueCred}{\ensuremath{\mathsf{IssueCred}}\xspace}
\newcommand{\VerifyCred}{\ensuremath{\mathsf{VerifyCred}}\xspace}

\newcommand{\PresentCred}{\ensuremath{\mathsf{PresentCred}}\xspace}

\newcommand{\MyVerify}{\ensuremath{\mathsf{Verify}}\xspace}

\newcommand{\zkSNARK}{\ensuremath{\mathsf{zkSNARK}}\xspace}

\newcommand{\ZPK}{\ensuremath{\mathsf{pk}}\xspace}
\newcommand{\ZVK}{\ensuremath{\mathsf{vk}}\xspace}

\newcommand{\Cred}{\ensuremath{\mathsf{Cred}}\xspace}

\newcommand{\Idx}{\ensuremath{\mathsf{idx}}\xspace}

\newcommand{\IS}{\ensuremath{\mathsf{IS}}\xspace}
\newcommand{\RL}{\ensuremath{\mathsf{RL}}\xspace}

\newcommand{\Auxiliary}{\ensuremath{\mathsf{aux}}\xspace}

\newcommand{\Holder}{\ensuremath{\mathcal{H}}\xspace}
\newcommand{\Issuer}{\ensuremath{\mathcal{I}}\xspace}
\newcommand{\Verifier}{\ensuremath{\mathcal{V}}\xspace}
\newcommand{\Policy}{\ensuremath{pol}\xspace}

\newcommand{\Adv}{\ensuremath{\mathcal{A}}\xspace}

\newcommand{\IPK}{\ensuremath{pk}\xspace}
\newcommand{\ISK}{\ensuremath{sk}\xspace}

\newcommand{\NIZK}{\ensuremath{\mathsf{ZKP}}\xspace}
\newcommand{\PT}{\ensuremath{\mathsf{pt}}\xspace}

\newcommand{\Verify}{\ensuremath{\mathsf{Verify}}\xspace}

\def\BibTeX{{\rm B\kern-.05em{\sc i\kern-.025em b}\kern-.08em
    T\kern-.1667em\lower.7ex\hbox{E}\kern-.125emX}}
\begin{document}
\title{Achieving Flexible and Secure Authentication with Strong Privacy in Decentralized Networks}

\author{Bin~Xie,
        Rui~Song,
        Xuyuan~Cai\\
        Department of Computing\\
The Hong Kong Polytechnic University
}


\maketitle

\begin{abstract}
Anonymous credentials (ACs) are a crucial cryptographic tool for privacy-preserving authentication in decentralized networks, allowing holders to prove eligibility without revealing their identity.
However, a major limitation of standard ACs is the disclosure of the issuer's identity, which can leak sensitive contextual information about the holder.
Issuer-hiding ACs address this by making a credential's origin indistinguishable among a set of approved issuers.
Despite this advancement, existing solutions suffer from practical limitations that hinder their deployment in decentralized environments: unflexible credential models that restrict issuer and holder autonomy, flawed revocation mechanisms that compromise security, and weak attribute hiding that fails to meet data minimization principles.
This paper introduces a new scheme called \Sysname to overcome these challenges.
We propose a flexible credential model that employs vector commitments with a padding strategy to unify credentials from heterogeneous issuers, enabling privacy-preserving authentication without enforcing a global static attribute set or verifier-defined policies.
Furthermore, we design a secure decentralized revocation mechanism where holders prove non-revocation by demonstrating their credential's hash lies within a gap in the issuer's sorted revocation list, effectively decoupling revocation checks from verifier policies while maintaining issuer anonymity.
\Sysname also strengthens attribute hiding by utilizing zk-SNARKs and vector commitments, allowing holders to prove statements about their attributes without disclosing the attributes themselves or the credential structure.
Security analysis and performance evaluations demonstrate its practical feasibility for decentralized networks, where presenting a credential can be finished in 1s.
\end{abstract}

\begin{IEEEkeywords}
authentication, anonymous credential, credential revocation, privacy protection, decentralized networks
\end{IEEEkeywords}

\section{Introduction}
\label{sec:introduction}

The rapid Internet evolution has enabled users to access diverse services and resources ubiquitously.
Authentication serves as the cornerstone for establishing trust between communicating entities~\cite{7811268,8721279}.
Currently, this process predominantly relies on Identity Providers (IdPs), such as Google and Facebook.
However, this centralized architecture invariably leads to the concentration of users' identity data and web footprints, exposing users to severe risks of large-scale privacy breaches and potential attacks targeting the availability of identity infrastructure.
Consequently, driven by the enforcement of privacy protection regulations, such as the General Data Protection Regulation (GDPR), establishing a reliable and privacy-preserving authentication mechanism in a decentralized network architecture has emerged as a research focus.


To address these concerns, decentralized identity systems introduce a tripartite model comprising issuers, holders, and verifiers.
By leveraging verifiable credentials and anonymous credentials (ACs)~\cite{10705058,8635539}, holders can prove eligibility to verifiers while minimizing attribute disclosure.
However, a significant privacy limitation persists.
Standard verification protocols necessitate the disclosure of the issuer's public key.
This requirement inadvertently reveals the issuer's identity, such as a specific university or local government, thereby enabling the malicious verifiers to infer the holder's affiliation or location.
Such correlation attacks fundamentally undermine the privacy guarantees intended by anonymous credential schemes.

\begin{figure*}[htbp]
    \centering
    \includegraphics[width=0.8\textwidth]{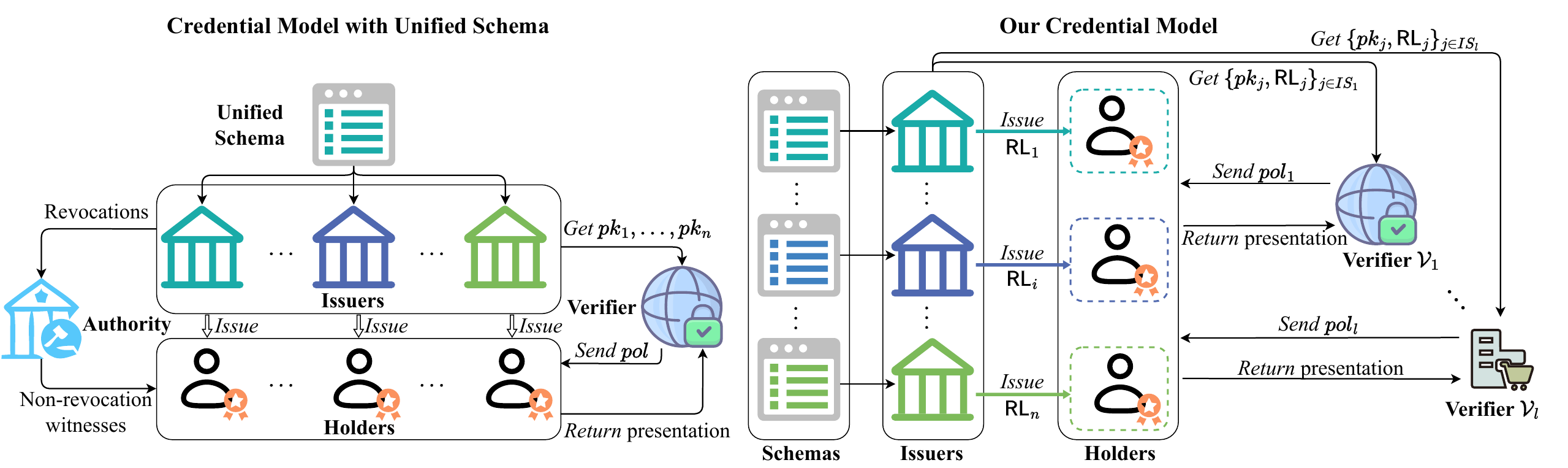}
    \caption{Unflexible credential model (left) and our new credential model (right).}
    \label{fig:model-comparison}
\end{figure*}

To resolve the above privacy concerns, one can employ \textit{issuer-hiding attribute-based credentials}, as formally introduced by Bobolz et al.~\cite{bobolz2021issuer}.
This technology allows a credential holder to produce a cryptographic token for authentication, demonstrating the credential's validity against a presentation policy that specifies an approved issuer set.
By design, the credential's specific origin, namely the issuer's public key, is rendered indistinguishable from those of other issuers in the designated set.
While several such credential systems~\cite{bobolz2021issuer,bosk2022hidden,mir2023aggregate,shi2023double,hui2025ring} have been developed, their practical application in decentralized networks is currently hindered by the following factors:
\begin{itemize}
	\item Unflexible credential model: Some schemes adopt one-attribute credentials~\cite{bosk2022hidden,mir2023aggregate}, causing high issuance overhead as each attribute demands separate identity check. Conversely, others~\cite{bobolz2021issuer,shi2023double,hui2025ring} mandate a uniform schema across all issuers (Fig.~\ref{fig:model-comparison}, left), restricting their autonomy to define context-specific attributes. This rigidity undermines the ``issue once, use everywhere'' paradigm, as a single schema cannot satisfy the diverse requirements of different verifiers. Furthermore, these schemes compel verifiers to pre-define rigid presentation policies, which limits holders' flexibility in issuer selection and complicates the verification workflow.
	\item Deficiencies in revocation: A credential revocation mechanism provides a means to invalidate credentials that have been compromised or whose conditions are no longer met. However, privacy-preserving revocation is a challenging feature to incorporate into issuer-hiding ACs because existing solutions expose the issuer's identity by requiring their unique revocation/allow lists as a public input during token verification~\cite{fueyo2016efficiency,camenisch2009accumulator}. A prior attempt to solve this issue~\cite{bobolz2021issuer} was found to be insecure: non-revocation witnesses retain validity even if the credential is revoked afterwards, permitting fraudulent non-revocation claims (see Section~\ref{subsec:strawman} for more details).
	\item Weak attribute hiding: Most issuer-hiding solutions often involve direct disclosure of the required attribute values~\cite{bosk2022hidden,mir2023aggregate,shi2023double}. Nevertheless, this practice is insufficient as it fails to meet the stringent requirements of attribute hiding and the data minimisation principle specified by laws such as GDPR. Attributes within a credential often encompass more information than is strictly necessary for verification, such as using a full birthdate for a simple age check. While these attributes individually may not be unique identifiers, they can act as quasi-identifiers. The combination of multiple such quasi-identifiers during an authentication process can uniquely pinpoint a holder, thus negating ACs' intended anonymity. 
\end{itemize}

The shortcomings of existing solutions motivate us to propose a new anonymous credential called \underline{I}ssuer-hiding and \underline{R}evocable \underline{A}nonymous \underline{C}redential (\Sysname) as sketched in Fig.~\ref{fig:model-comparison}(right). We now describe how \Sysname solves these issues in an efficient and decentralized way. 

To make the credential model flexible, we introduce an extensible attribute universe and holder-centric issuer set.
Unlike previous solutions that enforce a uniform schema, our approach allows each issuer to select a subset of attribute universe for issuance.
To conceal the differences in attribute sets and enable issuer-hiding, we employ a vector commitment scheme where issuers pad their assigned attribute values to a fixed-length vector using placeholders.
This produces a compact, fixed-size commitment for the credential, ensuring a unified format across heterogeneous issuers.
Besides, instead of requiring verifiers to define presentation policies, we empower holders to independently choose their issuer set during presentation generation and construct efficient credential presentations based on vector commitments.

To address the revocation status check, we propose a decentralized revocation mechanism where issuers maintain sorted lists of revoked credential hashes. Instead of relying on verifier-defined policies, the holder independently retrieves public keys and revocation lists from relevant issuers. The holder then aggregates these into a vector commitment and proves that their credential's hash lies between two adjacent elements in the issuer's committed sorted revocation list. This efficiently proves non-revocation without revealing the specific issuer or the full revocation list.

To achieve strong attribute hiding, we utilize the properties of vector commitments and zk-SNARKs. The holder proves knowledge of the credential commitment and a specific set of attribute index-value pairs required by the verifier.
By demonstrating that these pairs are valid elements within the committed vector and satisfy the verification predicate, the holder authenticates successfully without disclosing any unrequested attributes or specific attribute values in the credential, thus preserving strict data minimization.

In summary, we make the following contributions:

\begin{itemize}

    \item We propose \Sysname, a novel issuer-hiding credential scheme that supports an extensible attribute universe and holder-centric issuer selection. By employing vector commitments with a padding strategy, we unify credentials from heterogeneous issuers into a fixed-format commitment, enabling flexible and privacy-preserving authentication without enforcing a global static attribute set or verifier-defined policies.

    \item We design a secure decentralized revocation mechanism where holders prove non-revocation by demonstrating their credential's hash lies within a gap in the issuer's sorted revocation list. This approach effectively decouples revocation checks from verifier policies, prevents replay attacks, and maintains issuer anonymity.

    \item We incorporate strong attribute hiding capabilities using zk-SNARKs and vector commitments, allowing holders to prove predicate satisfaction without disclosing unrequested attributes or specific values. 
    
    \item We provide a rigorous security analysis and performance evaluation, demonstrating \Sysname's practicality for decentralized identity systems.
    
\end{itemize}

The rest of the paper is organized as follows. Section~\ref{sec:preliminary} introduces some preliminary knowledge. Section~\ref{sec:overview} gives an overview of \Sysname and is followed by the formal definition in Section~\ref{sec:sys-definition}. We then give a concrete design and security analysis in Section~\ref{sec:protocol}. Section~\ref{sec:evaluation} conducts the evaluation. Finally, we review the literature in Section~\ref{sec:related-work} and conclude this paper in Section~\ref{sec: conclusion}.

\section{Preliminary}
\label{sec:preliminary}

This section introduces the notations and cryptographic primitives utilized throughout this paper.

\subsection{Notations}
We use bold symbols such as $\vec{a}$ to denote vectors and $\vec{a}[i]$ to denote the $i$-th element of vector $\vec{a}$.
We use $||$ to denote the concatenation operation over two bit strings.
For a non-deterministic polynomial time (NP) relation $\mathcal{R}$, we write $(\mathsf{x};\mathsf{w})\in \mathcal{R}$ to indicate that $\mathsf{w}$ is a witness for the public statement $\mathsf{x}$. PPT stands for ``probabilistic polynomial time".

\subsection{Digital Signatures}
A digital signature scheme $\Sigma$ typically consists of the following three algorithms~\cite{10.1007/978-3-319-96881-0_4}:

\begin{itemize}
    \item $\MyPP\sample \Sigma.\mathsf{Setup}(1^\lambda) $: The setup algorithm takes a security parameter $1^\lambda$ as input and outputs public parameters $\MyPP$, which are implicit inputs to the following algorithms.
	\item $(sk,pk) \sample \Sigma.\mathsf{KeyGen}(1^\lambda)$: The key generation algorithm takes a security parameter $1^\lambda$ as input and outputs a secret/public key pair $(sk,pk)$.
	\item $\sigma\gets \Sigma.\mathsf{Sign}(sk,m)$: The signing algorithm takes the secret key $sk$ and a message $m$ as input and outputs a signature $\sigma$.
	\item $b\gets \Sigma.\mathsf{Verify}(pk,m,\sigma)$: The verification algorithm takes the public key $pk$, a message $m$, and a signature $\sigma$ as input, and outputs a bit $b$ indicating the verification result.
\end{itemize}

A secure digital signature scheme must be existentially unforgeable under chosen message attacks (EUF-CMA), even when the adversary is allowed to query a signing oracle on messages of its choice.

\begin{definition}[EUF-CMA]
	A digital signature scheme is EUF-CMA secure if and only if for any PPT adversary \Adv, there exists a negligible function $\mathsf{negl}$ such that $\Pr\left [ \mathsf{Exp}^{\mathsf{EUF-CMA}}_{\Adv}(\lambda)=1\right ]\leq \mathsf{negl}(\lambda)$, where the experiment $\mathsf{Exp}^{\mathsf{EUF-CMA}}_{\Adv}(\lambda)$ is defined in Fig.~\ref{fig:EUF-CMA}.
\end{definition}

\begin{figure}[htbp]
        \centering
        \begin{pcvstack}[boxed, center, space=1em]
            \procedure[linenumbering, width=0.45\columnwidth]{$\mathsf{Exp}^{\mathsf{EUF-CMA}}_{\Adv}(\lambda)$:}{
            \text{$\MyPP\sample \Sigma.\mathsf{Setup}(1^\lambda)$}\\
             \text{$(sk,pk)\sample \Sigma.\mathsf{KeyGen}(1^\lambda)$} \\
            \text{$\mathcal{Q}\gets \emptyset,(m^*,\sigma^*)\gets \Adv^{\mathcal{O}_{\mathsf{Sign}}}(pk)$}\\
            \text{return $1$ if $\Sigma.\mathsf{Verify}(pk,m^*,\sigma^*)=1 \wedge m^*\notin \mathcal{Q}$}\\
            \text{return $0$}
            }
            \procedure[width=0.4\columnwidth]{$\mathcal{O}_{\mathsf{Sign}}(m)$:}{
                \pcln \text{$\sigma\gets \Sigma.\mathsf{Sign}(sk,m),\mathcal{Q}\gets \mathcal{Q}\cup \{m\}$} \\ 
                \pcln \text{return $\sigma$} 
            }    
        \end{pcvstack}
        \caption{The EUF-CMA experiment for $\Sigma$.}
        \label{fig:EUF-CMA}
    \end{figure}

\subsection{zk-SNARK}

A zk-SNARK enables a prover to convince a verifier that it possesses certain information (i.e., a witness) without revealing any information about it~\cite{chen2023reviewzksnarks}. A zk-SNARK scheme $\Pi$ consists of the following algorithms:

\begin{itemize}
	\item $\MyPP \sample \Pi.\mathsf{Setup}(1^{\lambda})$: The setup algorithm takes a security parameter $1^\lambda$ as input and outputs public parameters $\MyPP$.
	\item $(\ZPK,\ZVK)\sample \Pi.\mathsf{KeyGen}(\MyPP,\mathcal{R})$: The key generation algorithm takes public parameters $\MyPP$ and an NP relation $\mathcal{R}$ as input, and outputs a proving/verifying key pair $(\ZPK,\ZVK)$.
	\item $\pi \gets \Pi.\mathsf{Prove}(\ZPK, \mathsf{x},\mathsf{w})$: The proving algorithm takes the proving key $\ZPK$, a statement $\mathsf{x}$, and a witness $\mathsf{w}$ as input, and outputs a succinct proof $\pi$.
	\item $b\gets \Pi.\mathsf{Verify}(\ZVK,\mathsf{x},\pi)$: The verification algorithm takes the verifying key $\ZVK$, a statement $\mathsf{x}$, and a proof $\pi$ as input, and outputs a bit $b\in \{0,1\}$ indicating the verification result.
\end{itemize}

A zk-SNARK scheme must satisfy the properties of \emph{completeness}, \emph{knowledge soundness}, \emph{zero-knowledge}, and \emph{succinctness}.

\begin{definition}[Perfect Completeness] 
A zk-SNARK scheme $\Pi$ is perfectly complete if for any PPT adversary \Adv, the following holds:
\begin{footnotesize}
\begin{equation}
\nonumber
    \Pr \left[
    \begin{array}{r}
        \MyPP\gets \Pi.\mathsf{Setup}(1^\lambda) \\
        (\mathcal{R}, (\mathsf{x}, \mathsf{w})) \gets \Adv(\MyPP) \\
        (\mathsf{ x}, \mathsf{w}) \in \mathcal{R} \\
        (\ZPK,\ZVK)\gets \Pi.\mathsf{KeyGen}(\mathcal{R})  \\
        \pi\gets \Pi.\mathsf{Prove}(\ZPK,\mathsf{x},\mathsf{w})
    \end{array}
    :
    \Pi.\mathsf{Verify}(\ZVK, \mathsf{x}, \pi) = 1 
    \right] = 1 
\end{equation}
\end{footnotesize}  
\end{definition}

\begin{definition}[Knowledge Soundness] 
A zk-SNARK scheme $\Pi$ satisfies knowledge soundness if for all PPT adversaries \Adv, there exists a PPT extractor $\mathcal{E}$ such that:
\begin{footnotesize}
\begin{equation}
 \nonumber
    \Pr \left[ \!\!\!
    \begin{array}{r}
    \MyPP\gets \Pi.\mathsf{Setup}(1^\lambda)  \\
        (\ZPK,\ZVK) \gets \Pi.\mathsf{KeyGen}(\mathcal{R})\\
        (\mathsf{x}, \pi)\gets \mathcal{A}(\MyPP,\ZPK,\ZVK,\mathcal{R}) \\
        \mathsf{w} \gets \mathcal{E}(\ZPK,\ZVK,\mathsf{x},\pi)
    \end{array}
    \!\!\! : \!\!\!
    \begin{array}{r}
        \Pi.\mathsf{Verify}(\ZVK, \mathsf{x}, \pi) = 1 \\
        (\mathsf{x}, \mathsf{w}) \notin \mathcal{R}
    \end{array}
    \!\!\! \right] \leq \mathsf{negl}(\lambda)
\end{equation}
\end{footnotesize}
\end{definition}

\begin{definition}[Succinctness] 
A zk-SNARK scheme $\Pi$ is succinct if the size of the proof $\pi$ is polylogarithmic in the size of the witness $w$.
\end{definition}

\begin{definition}[Zero-Knowledge] 
A zk-SNARK scheme $\Pi$ is zero-knowledge if for any NP relation $\mathcal{R}$ and $(\mathsf{x},\mathsf{w})\in \mathcal{R}$, there exists a PPT simulator $\mathcal{S}$ such that for all PPT adversaries \Adv, the following holds:
\begin{footnotesize}
\begin{align}
    &\Pr\left[
    \begin{array}{r}
        \MyPP \gets \Pi.\mathsf{Setup}(1^\lambda)\\
        (\ZPK,\ZVK)\gets \Pi.\mathsf{KeyGen}(\mathcal{R}) \\
        \pi \gets \Pi.\mathsf{Prove}(\ZPK,\mathsf{x},\mathsf{w})
    \end{array}
    :
    \Adv(\MyPP,(\ZPK,\ZVK),\pi)=1
    \right] \nonumber =  \\
    &\Pr\left[
    \begin{array}{r}
        \MyPP\gets \Pi.\mathsf{Setup}(1^\lambda) \\
        (\ZPK,\ZVK)\gets \Pi.\mathsf{KeyGen}(\mathcal{R}) \\
        \pi\gets \mathcal{S}(\MyPP,\mathcal{R},\mathsf{x})
    \end{array}
    :
    \Adv(\MyPP,(\ZPK,\ZVK),\pi)=1
    \right] \nonumber
\end{align}
\end{footnotesize}
\end{definition}

\subsection{Vector Commitment}

A vector commitment scheme $\Theta$ consists of the following algorithms~\cite{cryptoeprint:2025/667}:

\begin{itemize}
	\item $\MyPP \sample \Theta.\mathsf{Setup}(1^\lambda,n)$: The setup algorithm takes a security parameter $1^\lambda$ and the vector length $n$ as input, and outputs public parameters $\MyPP$ for the subsequent algorithms.
	\item $c \gets \Theta.\mathsf{Commit}(\vec{a})$: The commitment algorithm takes a vector $\vec{a}$ as input and outputs a commitment $c$.
	\item $\pi_i \gets \Theta.\mathsf{Open}(\vec{a},i,y)$: The opening algorithm takes a vector $\vec{a}$ and its $i$-th element $y$ as input and outputs a proof $\pi_i$.
	\item $b\gets \Theta.\mathsf{Verify}(c,i,y,\pi_i)$: The verification algorithm takes the commitment $c$, the $i$-th element $y$, and a proof $\pi_i$ as input, and outputs a bit $b=1$ (accept) or $b=0$ (reject).
\end{itemize}

The vector commitment scheme $\Theta$ must satisfy \emph{correctness} and \emph{binding}.

\begin{definition}[Correctness]
A vector commitment scheme $\Theta$ is correct if the following holds:
\begin{footnotesize}
	\begin{equation}
		\nonumber
		\Pr \left [
		\begin{array}{r}
			\MyPP \sample \Theta.\mathsf{Setup}(1^\lambda,n)\\
			c\gets \Theta.\mathsf{Commit}(\vec{a})\\
			\pi_i \gets \Theta.\mathsf{Open}(\vec{a},i,\vec{a}[i])
		\end{array}
		:
		\Theta.\mathsf{Verify}(c,i,\vec{a}[i],\pi_i)=1
		\right ]=1
	\end{equation}
\end{footnotesize}
\end{definition}

\begin{definition}[Binding]
A vector commitment scheme $\Theta$ is binding if, for $\MyPP \gets \Theta.\mathsf{Setup}(1^\lambda,n)$ and any PPT adversary $\adv$ that outputs $(c,i,y_0,y_1,\pi_0,\pi_1)$, the following holds:
\begin{footnotesize}
		\begin{equation}
			\nonumber
			\Pr \left[
			\begin{array}{r}
				y_0\neq y_1\\
				\Theta.\mathsf{Verify}(c,i,y_0,\pi_0)=1\\
				\Theta.\mathsf{Verify}(c,i,y_1,\pi_1)=1
			\end{array}
			\right]\leq \mathsf{negl}(\lambda)
		\end{equation}
	\end{footnotesize}
\end{definition}

\section{Overview}
\label{sec:overview}
This section first introduces the problem that \Sysname aims to address. Next, we present the baseline solution and discuss its limitations. Finally, we provide a high-level overview of how \Sysname overcomes these challenges.

\subsection{Problem Statement}

We consider the authentication problem in decentralized networks, where each entity may act as a credential issuer, holder, or verifier. An issuer $\Issuer_i$ issues a credential $\Cred_{i,j}$ to a holder $\Holder_j$ after verifying its identity. The holder $\Holder_j$ can then access the services of a target verifier $\Verifier_k$ if the attributes of $\Cred_{i,j}$ satisfy an access criterion specified by a predicate $\phi_{k}$ and the credential passes verification under the issuer's public key $\IPK_i$.

Beyond authentication security, we focus on privacy protection and credential revocation. The system should reveal nothing beyond the predicate $\phi_k$ and prove that the presented credential is not on the issuer's revocation list, all while maintaining anonymity. Additionally, the authentication process must ensure \emph{issuer hiding}, preventing unintended disclosure of the issuer's identity.

\subsection{Baseline Solution \& Limitations}
\label{subsec:strawman}

We begin by outlining the main idea of CANS'21~\cite{bobolz2021issuer}, which formally defines \emph{issuer hiding} and provides a concrete cryptographic construction for the problem described above. We then analyze its limitations in practical deployment, which this paper seeks to address.

\spar{Baseline Solution} The protocol in CANS'21 starts with a setup phase that initializes system parameters $\MyPP$, implicitly specifies an attribute set $\vec{\mathbb{A}}$, and selects a signature scheme $\Sigma_\Issuer$. Each issuer $\Issuer_i$ generates its public-secret key pair $(\ISK_i,\IPK_i)$ based on $\MyPP$ and $\Sigma_\Issuer$.

During credential issuance, $\Issuer_i$ assigns attribute values $\vec{a}_{i,j}$ according to $\vec{\mathbb{A}}$ for holder $\Holder_j$ and signs $(\vec{a}_{i,j},rh_{i,j},rpk_i)$, where $rh_{i,j}$ is a unique revocation handler and $(rsk_i,rpk_i)$ is a key pair held by $\Issuer_i$ using signature scheme $\Sigma_{\mathcal{R}}$. $\Issuer_i$ also creates two signatures $\sigma_{a}$ and $\sigma_{b}$ for $rh_{i,j}\in [a,b]$, where $[a,b]$ denotes the interval of unrevoked credentials.
The verifier $\Verifier_k$ defines a presentation policy $\Policy_k$ in three steps: \cirn{172} selects a set of issuers $\{\Issuer_i\}_{i\in \IS}$ where $\IS$ is an issuer index set; \cirn{173} generates a key pair $(vsk_v,vpk_v)$ using $\Sigma_\Verifier$; \cirn{174} signs each $\IPK_i$ in $\{\IPK_i\}_{i\in \IS}$ with $vsk_v$ to produce a presentation policy $\Policy_k=\{\sigma_{i,v}\}_{i\in \IS}$.

Given the credential $\Cred_{i,j}$ from $\Issuer_{i}$ and the presentation policy $\Policy_k$ from $\Verifier_k$, the holder $\Holder_j$ derives a credential presentation for predicate $\phi_k$ under a session identifier \Ctx using a non-interactive ZKP scheme as follows:
\begin{equation}
\nonumber
	\begin{aligned}
		\PT\gets&\NIZK\{((vpk_k,\phi_k);(\vec{a}_{i,j},rh_{i,j},rpk_i,\sigma_{i,j},\IPK_{i},\sigma_{i,v},\\
		&a,b,\sigma_a,\sigma_b)):\phi_k(\vec{a}_{i,j})=1\wedge a \leq rh_{i,j}\leq b\\
		& \wedge \Sigma_\Issuer.\Verify(\IPK_i,(\vec{a}_{i,j},rh_{i,j},rpk_i),\sigma_{i,j})=1\\
		&\wedge \Sigma_{\Verifier}.\Verify(vpk_k,\IPK_i,\sigma_{i,v})=1\\
		&\wedge \Sigma_{\mathcal{R}}.\Verify(rpk_i,a,\sigma_{a})=1\\
		&\wedge \Sigma_{\mathcal{R}}.\Verify(rpk_i,b,\sigma_{b})=1
		\}(\Ctx).
	\end{aligned}
\end{equation}

Utilizing \NIZK with \Ctx as the part of the input of the hash function used in Fiat-Shamir heuristic~\cite{10.1007/3-540-47721-7_12}, the holder can prove that its credential $\Cred_{i,j}$ is valid under a hidden public key $\IPK_i$ from an issuer in $\{\Issuer_i\}_{i\in \IS}$, without revealing attribute information beyond $\phi_k$, and that the credential remains unrevoked.

\spar{Limitations} The primary limitation is the requirement for global consensus on credential attributes among all issuers. In practice, issuers are heterogeneous and require support for customizable attributes rather than a static, universal set. This rigidity severely limits scalability: adding new attributes for authentication necessitates complete system re-initialization, which is impractical for large-scale deployments.

Secondly, the scheme requires each verifier to designate an issuer set and create a presentation policy. Given the dynamic nature of issuers in decentralized systems, verifiers must continuously update these policies, and holders must synchronize with the latest versions. This architecture imposes significant operational overhead and effectively centralizes control at the verifier, restricting the holder's flexibility in selecting a hiding set of issuers.

Finally, the revocation mechanism is vulnerable. Since the public key $rpk_i$ for revocation is embedded in the credential and used as a witness in \NIZK alongside $a$, $\sigma_a$, $b$, and $\sigma_b$, the scheme is susceptible to replay attacks. Specifically, even if a credential identified by $rh_{i,j}$ is revoked, a malicious holder can reuse previously obtained data to generate a valid presentation. A more robust design is needed to decouple revocation from credential issuance and enable issuer-hiding revocation checks during authentication.

\subsection{Technical Overview}
\label{subsec:technique}

\begin{figure*}[htbp]
    \centering
    \includegraphics[width=\textwidth]{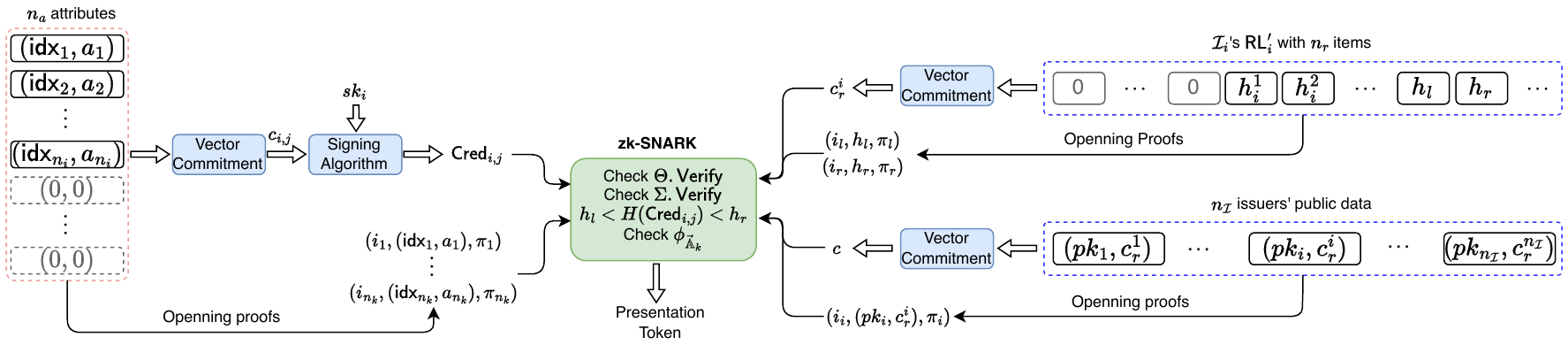}
    \caption{The workflow of \Sysname.}
    \label{fig:workflow}
\end{figure*}

In contrast to the baseline solution, which relies on a fixed and static attribute set, \Sysname introduces an extensible attribute universe $\mathbb{U}$.
This universe defines the set of identity attributes available for authentication and supports the dynamic addition of new attributes as the decentralized network evolves. Each attribute $\mathbb{A}_a$ is uniquely identified by an index $\mathsf{idx}_a$.
Accordingly, each issuer $\Issuer_i$ selects a subset of attributes $\vec{\mathbb{A}}_i$\footnote{We use $\vec{a}_{i,j}\in \vec{\mathbb{A}}_i$ to denote $\vec{a}_{i,j}$ is sampled from the value space of $\vec{\mathbb{A}}_i$.} from $\mathbb{U}$ and issues credentials to holders based on its chosen subset.

To address the heterogeneity of attribute sets among issuers and enable issuer-hiding authentication, it is necessary to construct a universal credential format that conceals differences in attribute cardinality.
To this end, $\Issuer_i$ pads its assigned attribute values to a fixed-length vector using placeholder values (0,0), where each meaningful element is an attribute index-value pair $(\mathsf{idx}_a, v_a)$.
The issuer then employs a vector commitment scheme $\Theta$ to commit to this vector, producing a compact, fixed-size commitment $c_a^{i,j}$, which is subsequently signed to issue a credential to $\Holder_j$.
For revocation, the issuer maintains a sorted list of revoked credential hash values, computed using a cryptographic hash function $H$.

During authentication, the holder $\Holder_j$ must demonstrate three things to the verifier $\Verifier_k$: 
\cirn{172} the attribute values satisfy the predicate $\phi_{\vec{\mathbb{A}}_v}$ which defines over a set of attributes $\vec{\mathbb{A}}_v$;
\cirn{173} the credential is valid under a public key $\IPK_{i}$ from a set of issuers;
\cirn{174} the credential has not been revoked by the corresponding issuer.

The first requirement is straightforward: $\Holder_j$ proves knowledge of a commitment $c_a^{i,j}$ and a set of attribute index-value pairs $\{(\mathsf{idx}_o, v_o)\}_{\mathbb{A}_o \in \vec{\mathbb{A}}_v}$ that are valid vector elements committed in $c_a^{i,j}$, and that these attribute values satisfy $\phi_{\vec{\mathbb{A}}_k}$. This can be efficiently realized using zk-SNARK and vector commitment proofs, without revealing information about any other attributes.

For the second and third requirements, we propose a unified proof approach to enhance scalability and eliminate the need for verifier-defined presentation policies. Instead, $\Holder_j$ independently retrieves a set of public keys $\{\IPK_i\}_{i \in \IS}$ and corresponding revocation lists $\{\RL_i\}_{i \in \IS}$ from issuers that endorse all required attributes, i.e., $\vec{\mathbb{A}}_v\subseteq \vec{\mathbb{A}}_i$ for all $i \in \IS$.
Analogous to attribute commitments, $\Holder_j$ uses $\Theta$ to commit to each padded revocation list $\RL^\prime_i$, generating a commitment $c_r^i$.
Subsequently, $\Holder_j$ aggregates $\{(\IPK_i, c_r^i)\}_{i \in \IS}$ into a vector commitment $c$ using $\Theta$.

Thus, $\Holder_j$ only needs to prove the existence of a tuple $(\IPK_i, c_r^i, \sigma_{i,j})$ and two elements $h_l, h_r$ such that: \cirn{172} $(\IPK_i, c_r^i)$ is committed in $c$; \cirn{173} $\IPK_i$ verifies $c_a^{i,j}$ and $\sigma_{i,j}$; \cirn{174} the hash value $H(\Cred_{i,j})$ lies between $h_l$ and $h_r$, where $h_l$ and $h_r$ are valid adjacent elements committed in $c_r^i$. Since $\RL_i^\prime$ is a sorted list (with padded placeholders set to the minimum possible value $0$), an element lies between two adjacent elements if and only if it is not included in the list, thereby confirming that the credential has not been revoked.

The above mechanisms enable \Sysname to support heterogeneous and extensible attribute sets, eliminate the need for presentation policies, and provide secure revocation checks without distributing non-revocation witnesses.
Fig.~\ref{fig:workflow} illustrates the workflow of \Sysname. The formal definition of \Sysname is deferred to Section~\ref{sec:sys-definition}, with further details provided in Section~\ref{sec:protocol}.



\section{Formal Definition}
\label{sec:sys-definition}
In this section, we present the formal definitions for the proposed \Sysname. We begin by describing the system model, which is followed by the syntax of \Sysname and its security definitions.

\subsection{System Model}

In \Sysname, the system consists of three primary types of entities that interact within the system:

\begin{itemize}
	\item \textbf{Issuers}: Trusted authorities responsible for performing identity checks and issuing credentials to holders. A distinguishing feature of issuers in \Sysname is their ability to support heterogeneous attribute sets.
	\item \textbf{Holders}: Entities that receive credentials from issuers and subsequently use them to authenticate to verifiers. They are free to choose which issuers to involve in each authentication session for issuer hiding, allowing for flexible and user-centric authentication processes.
	\item \textbf{Verifiers}: Entities that interact with holders to perform authentication. Verifiers must specify their access criteria in advance and determine the eligibility of a holder during the authentication process based on the credential presentation.
\end{itemize}

\spar{Adversary Model} We assume that issuers are trusted and honestly bind the attributes of holders into credentials. This assumption is justified by the design goal of \Sysname, which is to enable flexible, user-chosen issuer sets during authentication, and by the fact that issuers are not required to be aware of their involvement in any specific authentication session.

We consider two types of adversaries in the system:
\begin{itemize}
    \item \textbf{Malicious holders}: These adversaries attempt to subvert the authentication process by forging presentation tokens or bypassing honest verifiers, thereby compromising the security of the system.
    \item \textbf{Malicious verifiers}: These adversaries aim to extract more information from holders than is necessary for authentication, such as learning additional attribute values or linking multiple presentations to the same holder, thus violating privacy guarantees.
\end{itemize}

\subsection{Syntax}

Building on the analysis in Section~\ref{subsec:technique}, we formally define the syntax of \Sysname as a suite of algorithms and protocols that collectively realize the system's functionality.

\begin{definition}[\Sysname]
An \Sysname scheme is defined as a tuple of protocols $\Omega=(\MySetup,\AddAttributes, \IssuerSetup,\allowbreak \IssueCred,\Revoke,\PresentCred,\MyVerify)$, specified as follows:

\begin{itemize}
	\item $(\MyPP,\mathbb{U}) \gets \MySetup(1^\lambda,n_a,n_r,n_\Issuer)$: This setup algorithm takes as input a security parameter $1^\lambda$, the maximum number of attributes $n_a$, the maximum number of revocations $n_r$, the maximum size of a issuer set $n_\Issuer$. It outputs the system public parameters $\MyPP$ and an initial attribute universe $\mathbb{U}$, which serves as the global registry of supported attributes.
	\item $\mathbb{U}^\prime \gets \AddAttributes(\mathbb{U},\vec{\mathbb{A}})$: This algorithm updates the attribute universe by incorporating a new attribute set $\vec{\mathbb{A}}$, resulting in an updated universe $\mathbb{U}^\prime$ that reflects the evolving needs of the system.
	\item $(\ISK_i,\IPK_i,\vec{\mathbb{A}}_i,\Auxiliary_i) \gets \IssuerSetup(1^\lambda,\mathbb{U})$: Given the security parameter and the current attribute universe, this algorithm generates an issuer's secret key $\ISK_i$, public key $\IPK_i$, the subset of supported attributes $\vec{\mathbb{A}}_i \subseteq \mathbb{U}$, and initializes an empty revocation auxiliary data $\Auxiliary_i$.
	\item $\Cred_{i,j} \gets \IssueCred(\ISK_i,\vec{a}_{i,j},\vec{\mathbb{A}}_i)$: This credential issuance algorithm takes as input the issuer's secret key, a vector of attribute values $\vec{a}_{i,j}$ corresponding to $\vec{\mathbb{A}}_i$, and outputs a credential $\Cred_{i,j}$ for the holder.
	\item $b\gets \VerifyCred(\Cred_{i,j},\IPK_i,\vec{\mathbb{A}}_i)$: This verification algorithm takes as input a credential, the issuer's public key, and the subset of supported attributes, and outputs a bit $b$ indicating the validity of the credential.
	\item $\Auxiliary_i^\prime \gets \Revoke(\Auxiliary_i,\Cred_{i,j})$: This revocation algorithm updates the issuer's revocation auxiliary data $\Auxiliary_i$ by adding the specified credential $\Cred_{i,j}$, producing an updated auxiliary data $\Auxiliary_i^\prime$.
	\item $\PT \gets \PresentCred(\Cred_{i,j},\phi_{\vec{\mathbb{A}}_k},\Ctx,\{\IPK_p,\vec{\mathbb{A}}_p,\Auxiliary_p\}_{p \in \IS})$: The presentation algorithm allows a holder to generate a presentation token $\PT$ using their credential $\Cred_{i,j}$, a predicate $\phi_{\vec{\mathbb{A}}_k}$ that defines the access criteria, a context $\Ctx$ for the authentication session, and the relevant issuer information for the selected issuer set $\IS$.
	\item $b \gets \MyVerify(\PT, \phi_{\vec{\mathbb{A}}_k},\Ctx, \{\IPK_p,\vec{\mathbb{A}}_p,\Auxiliary_p\}_{p \in \IS})$: The verification algorithm takes as input a presentation token, the predicate, the context, and the relevant issuer information, outputting a bit $b$ that indicates whether the presentation is valid.
\end{itemize}
\end{definition}

\subsection{Security Definitions}

To rigorously capture the security guarantees of \Sysname, we define a set of oracles $(\mathcal{O}_{\mathsf{Add}},\mathcal{O}_{\mathsf{Issue}},\mathcal{O}_{\mathsf{Revoke}},\mathcal{O}_{\mathsf{Present}},\mathcal{O}_{\mathsf{Reveal}})$, as illustrated in Fig.~\ref{fig:oracles}. These oracles model the capabilities and knowledge accessible to an adversary during security experiments, reflecting realistic attack surfaces in the system.

\begin{figure}[htbp]
        \centering
        \begin{pcvstack}[boxed, center, space=1em]
            \procedure[linenumbering, width=0.45\columnwidth]{$\mathcal{O}_{\mathsf{Add}}(\vec{\mathbb{A}})$:}{
                \text{$\mathbb{U}^\prime \gets \AddAttributes(\mathbb{U},\vec{\mathbb{A}})$} \\
                \text{update $\mathbb{U}$ to $\mathbb{U}^\prime$}
            }
            \procedure[linenumbering, width=0.45\columnwidth]{$\mathcal{O}_{\mathsf{Issue}}(i,j,\vec{a}_{i,j})$:}{
            \text{$\Cred_{i,j}\gets \IssueCred(\ISK_i,\vec{a}_{i,j},\vec{\mathbb{A}}_i)$} \\
            \text{add $(i,j,\vec{a}_{i,j})$ to $\mathcal{Q}_{\mathsf{Issue}}$}
            }
            \procedure[linenumbering, width=0.2\columnwidth]{$\mathcal{O}_{\mathsf{Revoke}}(i,j)$:}{
                \text{add $\Cred_{i,j}$ to $\mathcal{Q}_{\mathsf{Revoke}}$} \\ 
                \text{return $\Auxiliary_i^\prime \gets \Revoke(\Auxiliary_i,\Cred_{i,j})$} 
            }
            \procedure[linenumbering, width=0.2\columnwidth]{$\mathcal{O}_{\mathsf{Present}}(i,j,\phi_{\vec{\mathbb{A}}_k},\Ctx,\IS)$:}{
                \text{add $(\phi_{\vec{\mathbb{A}}_k},\Ctx,\IS)$ to $\mathcal{Q}_{\mathsf{Present}}$} \\ 
                \text{$\PT \gets \PresentCred(\Cred_{i,j},\phi_{\vec{\mathbb{A}}_k},\Ctx,\{\IPK_i,\vec{\mathbb{A}}_i,\Auxiliary_i\}_{i \in \IS})$}\\
                \text{return $\PT$} 
            }
            \procedure[width=0.4\columnwidth]{$\mathcal{O}_{\mathsf{Reveal}}(i,j)$:}{
                \pcln \text{add $(i,j)$ to $\mathcal{Q}_{\mathsf{Reveal}}$} \\ 
                \pcln \text{return $\Cred_{i,j}$} 
            }    
        \end{pcvstack}
        \caption{The oracles that the adversary can access during security experiments.}
        \label{fig:oracles}
    \end{figure}

We now formalize the essential security properties that \Sysname must satisfy:

\emph{Correctness}: This property ensures that if all parties follow the protocol honestly, any presentation token generated by a holder will be accepted by the verifier. For brevity, we omit the formal definition, as it follows standard correctness requirements in cryptographic protocols.

\emph{Unforgeability}: This property guarantees that no malicious holder (i.e., adversary \Adv) can produce a valid presentation token that will be accepted by an honest verifier, except in cases where the adversary has previously obtained a unrevoked credential $\Cred_{i,j}$ satisfying $\phi_{\vec{\mathbb{A}}_k}$ or a presentation token for the same $(i,j,\phi_{\vec{\mathbb{A}}_k},\Ctx,\IS)$.  This ensures that credentials and presentations cannot be forged.

\begin{figure}[htbp]
        \centering
        \begin{pcvstack}[boxed, center, space=1em]
            \procedure[linenumbering, width=0.45\columnwidth]{$\mathsf{Exp}^{\mathsf{Unforgeability}}_{\Adv}(\lambda,n_a,n_r,n_\Issuer)$:}{
                \text{$(\MyPP,\mathbb{U}) \gets \Omega.\MySetup(1^\lambda,n_a,n_r,n_\Issuer)$} \\
                \text{$\mathcal{Q}_{\mathsf{Issue}} \gets \emptyset, \mathcal{Q}_{\mathsf{Present}} \gets \emptyset, \mathcal{Q}_{\mathsf{Reveal}} \gets \emptyset,\mathcal{Q}_{\mathsf{Revoke}} \gets \emptyset$}\\
                \text{$\mathbb{U} \gets \Adv^{\mathcal{O}_{\mathsf{Add}}}(\MyPP,\mathbb{U})$}\\
                \text{for each $i\in [n_\Issuer]$}:\\
                \quad \text{$(\ISK_i,\IPK_i,\vec{\mathbb{A}}_i,\Auxiliary_i)\gets \Omega.\IssuerSetup(1^\lambda,\mathbb{U})$ }\\
                \text{$(\PT^*,\phi_{\vec{\mathbb{A}}_k}^*,\Ctx^*,\IS^*)\gets \Adv^{\mathcal{O}_{\mathsf{Issue}},\mathcal{O}_{\mathsf{Present}},\mathcal{O}_{\mathsf{Reveal}},\mathcal{O}_{\mathsf{Revoke}}}(\MyPP,$}\\
                \text{$\mathbb{U},\{\IPK_i,\vec{\mathbb{A}}_i,\Auxiliary_i\}_{i\in [n_\Issuer]})$}\\
                \text{return $1$ if:}\\
                \quad \text{$\Omega.\MyVerify(\PT^*, \phi_{\vec{\mathbb{A}}_k}^*,\Ctx^*,\{\IPK_p,\vec{\mathbb{A}}_p,\Auxiliary_p\}_{p\in \IS^*})=1$}\\
                \quad \text{$\IS^* \subseteq [n_\Issuer]$}\\
                \quad \text{$(\phi_{\vec{\mathbb{A}}_k}^*,\Ctx^*,\IS^*)\notin \mathcal{Q}_{\mathsf{Present}}$}\\
                \quad \text{$\nexists (i,j)\in \mathcal{Q}_{\mathsf{Reveal}}$ such that:}\\
                \quad \quad \text{$\phi_{\vec{\mathbb{A}}_k}^*(\vec{a}_{i,j})=1$ and $\Cred_{i,j} \notin \mathcal{Q}_{\mathsf{Revoke}}$}\\
                \text{return $0$}
            }
            \procedure[width=0.4\columnwidth]{$\mathsf{Exp}^{\mathsf{Unlinkability}}_{\Adv}(\lambda,n_a,n_r,n_\Issuer)$:}{
                \pcln \text{$b \sample \{0,1\}$} \\ 
                \pcln \text{$(\MyPP,\mathbb{U}) \gets \Omega.\MySetup(1^\lambda,n_a,n_r,n_\Issuer)$}\\
                \pcln \text{$(\{\Cred_l,i_l\}_{l\in [2]},\{\IPK_p,\Auxiliary_p,\vec{\mathbb{A}}_p\}_{p\in [n_\Issuer]},\phi_{\vec{\mathbb{A}}_k},\Ctx,\IS) \gets $}\\ 
                \pcln \text{$\Adv^{\mathcal{O}_{\mathsf{Revoke}}}(\MyPP,\mathbb{U})$}\\
                \pcln \text{$\PT \gets \Omega.\PresentCred(\Cred_b, \phi_{\vec{\mathbb{A}}_k},\Ctx,\{\IPK_p,\vec{\mathbb{A}}_p,\Auxiliary_p\}_{p \in \IS})$}\\
                \pcln \text{$b^* \gets \Adv(\PT)$}\\
                \pcln \text{return $1$ if:}\\
                \pcln \quad \text{$b=b^*$}\\
                \pcln \quad \text{for each $l\in [2]$ that:}\\
                \pcln \quad \quad \text{$i_l \in [n_\Issuer]$}\\
                \pcln \quad \quad \text{$\Omega.\VerifyCred(\Cred_l, \IPK_{i_l}, \vec{\mathbb{A}}_{i_l})=1$}\\
                \pcln \quad \quad \text{$\phi_{\vec{\mathbb{A}}_k}(\vec{a}_l)=1$}\\
            	\pcln \quad \quad \text{$\Cred_l \notin \mathcal{Q}_{\mathsf{Revoke}}$}\\
	            \pcln \text{return $b^\prime \sample \{0,1\}$}
            }    
        \end{pcvstack}
        \caption{The games that capture the security properties of \Sysname.}
        \label{fig:games}
    \end{figure}

  \emph{Unlinkability}: This property ensures that even a malicious verifier (i.e., adversary \Adv) cannot link two presentations generated from the same holder, even if it has access to the underlying credentials. Unlinkability also implies selective disclosure, as any leakage of information beyond the predicate $\phi_{\vec{\mathbb{A}}_k}$ would enable the adversary to distinguish between presentations derived from different credentials.

\begin{definition}[Unforgeability]
\Sysname satisfies unforgeability if, for any PPT adversary \Adv and any numbers of issuers $l_\Issuer$ and verifiers $l_\Verifier$, there exists a negligible function $\mathsf{negl}$ such that:
$$\Pr\left [ \mathsf{Exp}^{\mathsf{Unforgeability}}_{\Adv}(\lambda,n_a,n_r,n_\Issuer) = 1 \right ]\leq \mathsf{negl}(\lambda),$$
where the experiment is defined in Fig.~\ref{fig:games}.
\end{definition}

\begin{definition}[Unlinkability]
\Sysname satisfies unlinkability if, for any PPT adversary \Adv, there exists a negligible function $\mathsf{negl}$ such that:
$$\left |\Pr\left [ \mathsf{Exp}^{\mathsf{Unlinkability}}_{\Adv}(\lambda,n_a,n_r,n_\Issuer)=1 \right ]-1/2 \right |\leq \mathsf{negl}(\lambda),$$
where the experiment is defined in Fig.~\ref{fig:games}.
\end{definition}

\section{The Protocol}
\label{sec:protocol}
This section details the proposed \Sysname protocol. We begin by presenting a generic construction in Section~\ref{subsec:construction}, followed by practical instantiations of the underlying building blocks in Section~\ref{subsec:instantiation}. Finally, we provide a formal security analysis in Section~\ref{subse:security-analysis}.

\subsection{Protocol Construction}
\label{subsec:construction}

Let $\Sigma_\Issuer=(\mathsf{Setup},\mathsf{KeyGen},\mathsf{Sign},\mathsf{Verify})$ denote a digital signature scheme. Let $\Theta_a$, $\Theta_r$, and $\Theta_\Issuer$ be vector commitment schemes comprising the tuple of algorithms $(\mathsf{Setup},\mathsf{Commit},\mathsf{Open},\mathsf{Verify})$. Furthermore, let $\Pi=(\mathsf{Setup},\mathsf{KeyGen},\mathsf{Prove},\mathsf{Verify})$ represent a general-purpose zk-SNARK scheme.

We define the NP relation $\mathcal{R}_{\phi_{\vec{\mathbb{A}}_k}}$ utilized within the zk-SNARK to generate the credential presentation as follows:
\begin{equation}
\nonumber
	\begin{aligned}
		\mathcal{R}_{\phi_{\vec{\mathbb{A}}_k}}=&\{((\phi_{\vec{\mathbb{A}}_k},c);((c_a^{i,j},\sigma_{i,j}),(i_i,\IPK_i,c_r^i,\pi_\Issuer^i),\\
		&\{i_p,\Idx_p,v_p,\pi_a^p\}_{\mathbb{A}_p \in \vec{\mathbb{A}}_k},(i_l,h_l,h_r,\pi_r^l,\pi_r^r))):\\
		& \wedge_{\mathbb{A}_p\in \vec{\mathbb{A}}_k} \Theta_a.\mathsf{Verify}(c_a^{i,j},i_p,(\Idx_p,v_p),\pi_a^p)=1\\
		&\wedge \phi_{\vec{\mathbb{A}}_k}(\{(\Idx_p,v_p)\}_{\mathbb{A}_p \in \vec{\mathbb{A}}_k})=1\\
		&\wedge \Sigma_\Issuer.\Verify(\IPK_i,c_a^{i,j},\sigma_{i,j})=1\\
		&\wedge \Theta_\Issuer.\mathsf{Verify}(c,i_i,(\IPK_i,c_r^i),\pi_\Issuer^i)=1\\
        &\wedge \Theta_r.\mathsf{Verify}(c_r^i,i_l,h_l,\pi_r^l)=1\\
        &\wedge \Theta_r.\mathsf{Verify}(c_r^i,i_l+1,h_r,\pi_r^r)=1\\
        &\wedge h_l <  H(c_a^{i,j}||\sigma_{i,j}) < h_r\}.
	\end{aligned}
\end{equation}

Building upon these cryptographic primitives and the defined NP relation, the generic construction of \Sysname is illustrated in Fig.~\ref{fig:constructions}.

\begin{figure*}[htbp]
        \centering
        \begin{pchstack}[boxed, center, space=1em]
             
             \begin{pcvstack}
                \procedure[linenumbering, width=0.24\columnwidth]{$\MySetup(1^\lambda,n_a,n_r,n_\Issuer)\rightarrow (\MyPP,\mathbb{U})$:}{
                    \text{Run $\MyPP_a \gets \Theta_a.\MySetup(1^\lambda,n_a)$;} \\
                    \text{Run $\MyPP_r \gets \Theta_r.\MySetup(1^\lambda,n_r)$;} \\
                    \text{Run $\MyPP_\Issuer \gets \Theta_\Issuer.\MySetup(1^\lambda,n_\Issuer);$} \\
                    \text{Run $\MyPP_z \gets \Pi.\MySetup(1^\lambda)$;} \\
                    \text{Run $\MyPP_s \gets \Sigma_\Issuer.\MySetup(1^\lambda)$;} \\
                    \text{Initialize $\mathbb{U}:=\emptyset$ and choose a hash function $H$;} \\
                    \text{Set $\MyPP:=(\MyPP_a,\MyPP_r,\MyPP_\Issuer,\MyPP_z,\MyPP_s,H)$;}\\
                    \text{Return $(\MyPP,\mathbb{U}).$}
                }
                \vspace{10pt}
                 \procedure[linenumbering, width=0.24\columnwidth]{$\AddAttributes(\mathbb{U},\vec{\mathbb{A}})\rightarrow \mathbb{U}^\prime$:}{
                    \text{Append all $\mathbb{A}\in \vec{\mathbb{A}}$ to $\mathbb{U}$;} \\
                    \text{Return the updated $\mathbb{U}$ as $\mathbb{U}^\prime$.}
                }
                \vspace{10pt}
                \procedure[linenumbering, width=0.24\columnwidth]{$\IssuerSetup(1^\lambda,\mathbb{U})\rightarrow (\ISK_i,\IPK_i,\vec{\mathbb{A}}_i,\Auxiliary_i)$:}{
                    \text{Run $(\ISK_i,\IPK_i) \gets \Sigma_\Issuer.\mathsf{KeyGen}(1^\lambda)$;} \\
                    \text{Select $\vec{\mathbb{A}}_i \subseteq \mathbb{U}$ and initialize $\Auxiliary_i \gets \emptyset$;} \\
                    \text{Return $(\ISK_i,\IPK_i,\vec{\mathbb{A}}_i,\Auxiliary_i)$.}
                }
                \vspace{10pt}
                \procedure[linenumbering, width=0.24\columnwidth]{$\IssueCred(\ISK_i,\vec{a}_{i,j},\vec{\mathbb{A}}_i)\rightarrow \Cred_{i,j}$:}{
                \text{Assert that $\vec{a}_{i,j}\in \vec{\mathbb{A}}_i$;} \\
                \text{Convert $\vec{a}_{i,j}$ into $\vec{a}^\prime_{i,j}$ with $\vec{a}^\prime_{i,j}[p] = (\Idx_{\vec{\mathbb{A}}_i[p]},\vec{a}_{i,j}[p])$;} \\
                \text{Pad $\vec{a}^\prime_{i,j}$ into $\vec{a}^{\prime\prime}_{i,j}$ with $(0,0)$ to a fixed length $n_a$;} \\
                \text{Run $c_a^{i,j}\gets \Theta_a.\mathsf{Commit}(\vec{a}^{\prime\prime}_{i,j})$;} \\
                \text{Run $\sigma_{i,j}\gets \Sigma_\Issuer.\mathsf{Sign}(\ISK_i,c_a^{i,j})$;} \\
                \text{Set $\Cred_{i,j}:=(\vec{a}_{i,j},\sigma_{i,j})$ and return it.}
                }
                \vspace{10pt}
                \procedure[linenumbering, width=0.24\columnwidth]{$\VerifyCred(\Cred_{i,j},\IPK_i,\vec{\mathbb{A}}_i)\rightarrow b$:}{
                    \text{Parse $\Cred_{i,j}$ as $(\vec{a}_{i,j},\sigma_{i,j})$;} \\
                    \text{Assert that $\vec{a}_{i,j}\in \vec{\mathbb{A}}_i$;} \\
                    \text{Reconstruct $c_a^{i,j}$ from $\vec{a}_{i,j}$;} \\
                    \text{Return $b\gets \Sigma_\Issuer.\mathsf{Verify}(\IPK_i,c_a^{i,j},\sigma_{i,j})$.}
                }

             \end{pcvstack}

              \begin{pcvstack}
               \procedure[linenumbering, width=0.24\columnwidth]{$\Revoke(\Auxiliary_i,\Cred_{i,j})\rightarrow \Auxiliary_i^\prime$:}{
                    \text{Calculte $h_{i,j}:=H(\Cred_{i,j})$}\\
                    \text{Parse $\Auxiliary_i$ into $(\RL_i,c_r^i)$}\\
                    \text{Add $h_{i,j}$ to $\RL_i$ and sort $\RL_i$ into $\RL_i^\prime$;}\\
                    \text{Pad $\RL_i^\prime$ with $0$ to a fixed length $n_r$;}\\
                    \text{Convert $\RL_i^\prime$ into a vector $\vec{L}_i$;}\\
                    \text{Obtain new $c_r^i\gets \Theta_r.\mathsf{Commit}(\vec{L}_i)$;}\\
                    \text{Return $\Auxiliary_i^\prime=(\RL_i^\prime,c_r^i)$.}
                }
                \vspace{10pt}
                 \procedure[width=0.24\columnwidth]{$\PresentCred(\Cred_{i,j},\phi_{\vec{\mathbb{A}}_k},\Ctx,\{\IPK_p,\vec{\mathbb{A}}_p,\Auxiliary_p\}_{p \in \IS})\rightarrow \PT$:}{
                    \pcln \text{Parse $\Cred_{i,j}$ as $(\vec{a}_{i,j},\sigma_{i,j})$;} \\
                    \pcln \text{Assert that $i\in \IS$ and $\vec{\mathbb{A}}_{k}\subseteq \vec{\mathbb{A}}_p$ for each $p\in \IS$;} \\
                    \pcln \text{Reconstruct $c_a^{i,j}$ from $\vec{a}_{i,j}$;} \\
                    \pcln \text{Run $\pi_a^p \gets \Theta_a.\mathsf{Open}(c_a^{i,j},i_p,(\Idx_p,v_p))$ for each $\mathbb{A}_p\in \vec{\mathbb{A}}_k$;}\\
                    \pcln \text{Obtain $c_r^p$ from $\Auxiliary_p$ for each $p\in \IS$;}\\
                    \pcln \text{Organize $\{(\IPK_p,c_r^p)\}_{p\in \IS}$ into a vector $\vec{C}$;}\\
                    \pcln \text{Run $c \gets \Theta_\Issuer.\mathsf{Commit}(\vec{C})$;}\\
                    \pcln \text{Run $\pi_\Issuer^i \gets \Theta_\Issuer.\mathsf{Open}(c,i_i,(\IPK_i,c_r^i))$ for $\Issuer_i$ that issues $\Cred_{i,j}$};\\
                    \pcln \text{Find $h_l$ and $h_r$ in $\RL_i^\prime$ such that $h_l < H(\Cred_{i,j}) < h_r$};\\
                    \pcln \text{Run $\pi_r^l \gets \Theta_r.\mathsf{Open}(c_r^i,i_l,h_l)$};\\
                    \pcln \text{Run $\pi_r^r \gets \Theta_r.\mathsf{Open}(c_r^i,i_r,h_r)$};\\
                    \pcln \text{Run $(\ZPK,\ZVK)\gets \Pi.\mathsf{KeyGen}(\MyPP_z,\mathcal{R}_{\phi_{\vec{\mathbb{A}}_k}})$;}\\
                    \pcln \text{Run $\pi \gets \Pi.\mathsf{Prove}(\ZPK,(\phi_{\vec{\mathbb{A}}_k},c),((c_a^{i,j},\sigma_{i,j}),(i_i,\IPK_i,c_r^i,\pi_\Issuer^i),$}\\
                    \text{ \quad \quad $\{i_p,\Idx_p,v_p,\pi_a^p\}_{\mathbb{A}_p \in \vec{\mathbb{A}}_k},(i_l,h_l,h_r,\pi_r^l,\pi_r^r))$ with \Ctx;}\\
                    \pcln \text{Return $\PT:=\pi$.}
                }    
                \vspace{10pt}
                 \procedure[width=0.24\columnwidth]{$\Verify(\PT,\phi_{\vec{\mathbb{A}}_k},\Ctx,\{\IPK_p,\vec{\mathbb{A}}_p,\Auxiliary_p\}_{p \in \IS})\rightarrow b$:}{
                    \pcln \text{Assert $\vec{\mathbb{A}}_{k}\subseteq \vec{\mathbb{A}}_p$ for each $p\in \IS$; }\\
                    \pcln \text{Reconstruct $c$ according to $\{\IPK_p,\Auxiliary_p\}_{p \in \IS};$} \\
                    \pcln \text{Run $(\ZPK,\ZVK)\gets \Pi.\mathsf{KeyGen}(\MyPP_z,\mathcal{R}_{\phi_{\vec{\mathbb{A}}_k}})$;}\\
                    \pcln \text{Return $b\gets \Pi.\mathsf{Verify}(\ZPK,(\phi_{\vec{\mathbb{A}}_k},c),\PT)$ with \Ctx.} 
                } 
             \end{pcvstack}

    \end{pchstack}
        \caption{The generic construction of \Sysname.}
        \label{fig:constructions}
    \end{figure*}

\subsection{Instantiation}
\label{subsec:instantiation}
This section discusses the practical instantiation of the building blocks within \Sysname. Given that \Sysname relies on zk-SNARKs for credential presentation, selecting zk-friendly primitives is crucial for optimizing system efficiency.

Specifically, the NP relation $\mathcal{R}_{\phi_{\vec{\mathbb{A}}_k}}$ primarily involves the following operations: (1) vector commitment verification; (2) digital signature verification; (3) hash computation; and (4) predicate evaluation. To minimize proving time, we select efficient constructions for these operations.

For vector commitments, we employ the Merkle tree-based scheme proposed in~\cite{cryptoeprint:2025/667}, which offers efficient commitment and opening algorithms. In conjunction with zk-friendly hash functions such as Poseidon~\cite{poseidon}, we achieve efficient in-circuit hash computations.

For digital signatures, we adopt standard schemes such as ECDSA, Schnorr, and EdDSA. However, the group point multiplications required for verification are computationally expensive within a zk-SNARK circuit. To address this, we utilize the ``cycle of curves'' technique~\cite{cryptoeprint:2023/1192}, which transforms group point multiplications into field multiplications over the SNARK proving curve's scalar field, requiring only approximately 2.5k constraints to verify an ECDSA signature.

Finally, for predicate evaluation, we employ arithmetic circuits to implement constraints such as range proofs, set membership, and linear inequalities, provided the attribute values are encoded within the scalar field of the zk-SNARK's proving curve.

With these optimizations, we achieve an efficient instantiation of \Sysname, ensuring practical performance as demonstrated in Section~\ref{sec:evaluation}.

\subsection{Security Analysis}
\label{subse:security-analysis}

In this section, we provide the security proof for \Sysname.

\begin{theorem}[Unforgeability]
\label{theorem:unforgeability}
If the signature scheme $\Sigma_\Issuer$ is EUF-CMA secure, the vector commitment schemes satisfy the binding property, and the \zkSNARK scheme $\Pi$ satisfies zero-knowledge and knowledge soundness, then \Sysname satisfies unforgeability.
\end{theorem}

Broadly speaking, an adversary can break unforgeability in two ways: (1) by forging a $\Sigma_\Issuer$ signature (implying a compromise of an honest issuer's secret key), or (2) by forging vector commitment proofs (breaking the binding property of $\Theta_a$, $\Theta_r$, or $\Theta_\Issuer$).

\begin{proof}
Let \Adv be a PPT adversary. We first modify the unforgeability game by simulating all presentation tokens \PT output by $\mathcal{O}_{\mathsf{Present}}$.
Since the tokens are generated by a \zkSNARK that satisfies zero-knowledge, the adversary's advantage in the modified game increases by at most a negligible amount.

In this modified game, if \Adv wins, we can apply the \zkSNARK extractor to recover the witness $w = (((c_a^{i,j},\sigma_{i,j}),\allowbreak (i_i,\IPK_i,c_r^i,\pi_\Issuer^i),\{i_p,\Idx_p,v_p,\pi_a^p\}_{\mathbb{A}_p \in \vec{\mathbb{A}}_k},(i_l,h_l,h_r,\pi_r^l,\pi_r^r)))$.
We define two events: $E_e^f$ denotes the event where \Adv wins but the extractor fails to output a valid witness; $E_e^s$ denotes the event where \Adv wins, extraction succeeds, and the witness is valid.
It holds that $\Pr[\Adv\,\,\mathrm{wins}]=\Pr[E_e^f]+\Pr[E_e^s]$. Because $\Pi$ satisfies knowledge soundness, $\Pr[E_e^f]$ is negligible. We now demonstrate that $\Pr[E_e^s]$ is negligible, which concludes the proof.

The event $E_e^s$ implies that the extracted witness is valid and satisfies the constraints defined in relation $\mathcal{R}_{\phi_{\vec{\mathbb{A}}_k}}$.
We categorize the possible scenarios for $E_e^s$ into three cases and show that each occurs with negligible probability:
\begin{itemize}
    \item {Signature forgery or binding breach ($\Sigma_\Issuer, \Theta_\Issuer$):} The witness contains a valid signature $(\IPK_i,c_a^{i,j},\sigma_{i,j})$. If $\IPK_i$ corresponds to an honest issuer's public key committed in $c$, then \Adv has successfully forged a signature under $\Sigma_\Issuer$, contradicting its EUF-CMA security. Conversely, if $\IPK_i$ does not belong to any honest issuer but is accepted via the proof $\pi_\Issuer^i$ against the honest commitment $c$, then \Adv has forged a vector commitment opening proof under $\Theta_\Issuer$, contradicting the binding property of $\Theta_\Issuer$.
    \item {Attribute binding breach ($\Theta_a$):} The condition $\phi_{\vec{\mathbb{A}}_k}(\{\Idx_p,v_p\}_{\mathbb{A}_p\in \mathbb{A}_k})=1$ holds. This implies that the adversary has constructed a valid presentation satisfying the predicate without possessing a valid credential with those specific attributes. To achieve this, \Adv must forge the necessary attributes and associated vector commitment proofs with respect to the committed $c_a^{i,j}$, which contradicts the binding property of $\Theta_a$.
    \item {Revocation binding breach ($\Theta_r$):} The condition $h_l <  H(\Cred_{i,j}) < h_r$ holds. This implies that the adversary has successfully constructed a valid presentation for a revoked credential (where the hash should ideally be present in the list, not between elements). To achieve this, \Adv must forge two vector commitment proofs under $\Theta_r$ for $h_l$ and $h_r$ relative to $c_r^i$, which contradicts the binding property of $\Theta_r$.
\end{itemize}

Based on the above analysis, we conclude that $\Pr[E_e^s]$ is negligible. Thus, the overall advantage of \Adv in breaking unforgeability is negligible.
\end{proof}

\begin{theorem}[Unlinkability]
	If $\Pi$ satisfies the zero-knowledge property, then \Sysname satisfies unlinkability.
\end{theorem}

\begin{proof}
We provide a proof sketch, as this property follows directly from the zero-knowledge property of the zk-SNARK. Specifically, the unlinkability experiment is designed such that the witnesses used for computing the challenge presentation $\PT^*$ are valid regardless of whether the bit $b = 0$ or $b = 1$. Consequently, the real experiment is computationally indistinguishable from a simulated experiment where the \zkSNARK simulator generates the proofs. In the simulated view, the presentation token is independent of the witness and thus independent of the bit $b$. Therefore, any adversary capable of distinguishing $b$ with non-negligible advantage would violate the zero-knowledge property of $\Pi$.
\end{proof}

\section{Evaluation}
\label{sec:evaluation}

In this section, we present the experimental evaluation of \Sysname, demonstrating its efficiency and practicality.

\subsection{Experimental Configuration}
\emph{Experimental Environment}: We conducted all experiments on a Mac mini workstation equipped with an Apple M4 chip and 16 GB of RAM.

\emph{Implementation Details}: We set the security parameter $\lambda=128$ and adopted Plonk~\cite{cryptoeprint:2019/953} as the underlying \zkSNARK scheme.
Our implementation utilizes the Halo2 library\footnote{https://github.com/privacy-ethereum/halo2}, employing the BN254 curve for the \zkSNARK.
For credential issuance, we use the ECDSA signature scheme constructed over the Grumpkin curve.
Additionally, we employ the ZK-friendly Poseidon hash function~\cite{grassi2021poseidon} for vector commitments and signatures to optimize performance.

\emph{Baseline}: While several anonymous credential schemes offer issuer hiding and high efficiency, they typically lack support for revocation, a critical feature that significantly impacts system performance.
To ensure a fair and comprehensive comparison, we evaluate \Sysname against the most relevant scheme proposed in CANS'21~\cite{bobolz2021issuer}, which supports both issuer hiding and revocation.

\subsection{Experimental Results}

\spar{Performance of main operations} We first evaluate the computational costs of the primary operations in \Sysname, including system setup, attribute addition, issuer setup, credential issuance/verification/revocation, and credential presentation generation/verification.

The system setup phase involves generating system parameters, a process dominated by the generation of public parameters for the \zkSNARK.
Fig.~\ref{fig:evaluation} (left) illustrates the time cost of system setup across different maximum circuit constraints.
As observed, the setup time scales with the number of constraints.
With a maximum circuit constraint of $2^{16}$, the setup time is approximately 3.5 seconds.

The attribute addition operation is a lightweight task that simply appends new attribute information to the attribute universe; its computational cost is negligible.
The issuer setup operation primarily involves key pair generation and revocation list initialization. Its cost is largely determined by public key generation, requiring approximately 1.7 ms.
During credential issuance, the issuer generates a vector commitment for the attribute values and signs this commitment.
Theoretically, the issuance cost scales linearly with the number of attributes in the credential.
However, due to the high efficiency of the underlying hash function, the cost is dominated by signature generation, taking about 1.8 ms.
Similarly, the verification cost is dominated by signature verification, taking approximately 3.8 ms.
For revocation, the issuer hashes the revoked credential and inserts the result into a sorted revocation list.
This operation is highly efficient, costing less than 8 seconds, even with a revocation list containing $2^{15}$ entries, which can be further optimized via parallel computation.
We summarize the performance of these operations in Table~\ref{table:main-operations}.

\begin{table}
    \centering
    \caption{Performance of Main Operations} 
    \label{table:main-operations}
    \begin{tabular}{cccc}
        \toprule     
        \IssuerSetup & \IssueCred  &  \VerifyCred & \Revoke \\
        \midrule   
        1.7 ms &1.8 ms &3.8 ms & 7.2 s \\
        \bottomrule 
    \end{tabular}
\end{table}

\begin{figure}[!t]
	\centering
    \begin{minipage}{0.45\columnwidth}
        \resizebox{\textwidth}{!}{\begin{tikzpicture}
\begin{loglogaxis}[%
    width = 4in,
    height = 3.6in,
    xmode = linear,
    ymode = linear,
    log ticks with fixed point,
    xmin = 12,
    xmax = 16,
    ymin = 0,
    ymax = 4.2,
    xmajorgrids,
    ymajorgrids,
    xlabel={\# constraints (in $\log_2$ scale)},
    ylabel={Time/s},
    legend style={%
        font = \large\itshape,
        legend cell align = left,
        align = left,
    },
    label style = {
        font = \large,
    },
    legend pos=north west,
    tick label style = {
        scale = 1.5
    },
    xtick={12,13,14,15,16},
    ytick={1,2,3,4}
]

\addplot[
    color=momo,
    line width=2pt,
    mark size=3pt,
    mark=*,
    mark options={
        solid,
        momo,
        fill=white
    },
    smooth
]
coordinates {
(12,0.25)
(13,0.496)
(14,0.925)
(15,1.8)
(16,3.5)
};
\addlegendentry{\MySetup}

\end{loglogaxis}
\end{tikzpicture}}
    \end{minipage}
    \begin{minipage}{0.45\columnwidth}
        \resizebox{\textwidth}{!}{\begin{tikzpicture}
\begin{loglogaxis}[%
    width = 4in,
    height = 3.6in,
    xmode = linear,
    ymode = linear,
    log ticks with fixed point,
    xmin = 12,
    xmax = 16,
    ymin = -0.1,
    ymax = 1.25,
    xmajorgrids,
    ymajorgrids,
    xlabel={\# constraints (in $\log_2$ scale)},
    ylabel={Time/s},
    legend style={%
        font = \large\itshape,
        legend cell align = left,
        align = left,
    },
    label style = {
        font = \large,
    },
    legend pos=north west,
    tick label style = {
        scale = 1.5
    },
    xtick={12,13,14,15,16},
    ytick={0,0.25,0.5,0.75,1.0}
]

\addplot[
    color=momo,
    line width=2pt,
    mark size=2.5pt,
    mark=*,
    mark options={
        solid,
        momo,
        fill=white,
    },
    smooth
]
coordinates {
  (12,0.141)
  (13,0.221)
  (14,0.375)
  (15,0.67)
  (16,1.18)
};
\addlegendentry{\PresentCred}

\addplot[
    color=ruri,
    line width=2pt,
    mark size=2.5pt,
    mark=*,
    mark options={
        solid,
        ruri,
        fill=white
    },
    smooth
]
coordinates {
  (12,0.002)
  (13,0.002)
  (14,0.002)
  (15,0.002)
  (16,0.002)
};
\addlegendentry{\Verify}

\end{loglogaxis}
\end{tikzpicture}}
    \end{minipage}
    \caption{Performance results for system setup (left) and credential presentation/verification (right).}
    \label{fig:evaluation}
\end{figure}
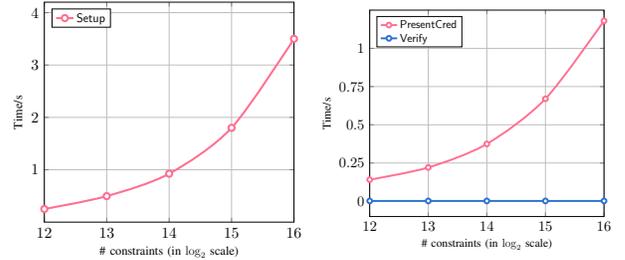

When generating a presentation, the holder collects the required information and generates the corresponding proof. This process is dominated by the \zkSNARK proving time.
The proving cost increases with the size of the issuer's revocation list, the issuer set size, and the complexity of the attribute predicates.
Leveraging vector commitments, the proving cost grows by a factor of $\log(n)$, where $n$ represents the size of the revocation list and the issuer set.
Given a revocation list with $2^{15}$ entries per issuer, an issuer set size of $2^{10}$, and a complex predicate within $2^{14}$ constraints, the total number of constraints for presenting a credential remains below $2^{15}$, which can be processed in 0.7 s.
The verification cost is dominated by the \zkSNARK verification process and remains constant at approximately 2 ms, regardless of circuit size.
Fig.~\ref{fig:evaluation} (right) presents the performance results for credential presentation and verification under different proving constraints.

\begin{center}
    \begin{threeparttable}
    \centering
    \caption{Comparison Results with the Baseline}
    \label{table:comparison}
	\setlength{\tabcolsep}{15pt}
    \begin{tabular}{cccc}
    \toprule     
        \makecell[c]{{Schemes}}&\makecell[c]{{Predicate \cirn{172}}} & \makecell[c]{{Predicate \cirn{173}}} \\ 
    \midrule 
        \makecell[c]{CANS'21} & 74 ms / 95 ms & 13 s / 127 ms \\
        Ours & 375 ms / 2 ms & 442 ms / 2 ms \\
    \bottomrule    
    \end{tabular}
    \begin{tablenotes}
    \scriptsize
        \item The values denote presentation generation time / verification time.
    \end{tablenotes}
\end{threeparttable}
\end{center}

\spar{Comparison with the baseline} We compare the presentation and verification costs of \Sysname with the CANS'21 baseline scheme across two different predicates.
We configure \Sysname with $2^7$ attributes per credential, $2^{15}$ revocation entries per issuer, and a hidden issuer set size of $2^{7}$.
We construct two distinct predicate instances for presentation and verification: \cirn{172} Age $> 18$, and \cirn{173} the holder is not present in a banlist of $2^{15}$ entries.
For the CANS'21 implementation, we employ the technique from~\cite{camenisch2008efficient} for the range proof in predicate \cirn{172} and the accumulator-based non-membership proof from~\cite{srinivasan2022batching} for predicate \cirn{173}.

Table~\ref{table:comparison} details the presentation and verification costs for both schemes.
As shown, \Sysname achieves performance comparable to CANS'21 for both predicates.
For predicate \cirn{172}, our scheme requires 375 ms for presentation generation and 2 ms for verification, whereas CANS'21 requires 74 ms and 95 ms, respectively.
For predicate \cirn{173}, our scheme takes 442 ms for presentation and 2 ms for verification, significantly outperforming CANS'21, which requires 13 seconds and 127 ms, respectively.
These results demonstrate that \Sysname achieves efficient credential presentation and verification while supporting secure revocation and flexible credential issuance.

\section{Related Work}
\label{sec:related-work}
\subsection{Anonymous Credentials}

Camenisch and Lysyanskaya~\cite{camenisch2001efficient} first proposed a fully anonymous credential scheme in 2001, which used the ZKP to show the owner has a pseudonym and a signature on it without revealing the real identity. Similar to Camenisch's work, Belenkiy~\cite{belenkiy2008p} proposed P-signatures based on a commitment and a signature scheme. It also used the ZKP to show knowledge of the signed value to achieve anonymity and unlinkability. Zhang et al.~\cite{zhang2012efficient} constructed anonymous credentials from attribute-based signatures and provided an optimized attribute proof. Beyond employing ZKP to achieve anonymity, some studies have proposed the construction and utilization of randomizable signatures. Verheul~\cite{verheul2001self} introduced an anonymous credential scheme that ensures unlinkability across sessions through credential randomization before presentation. Hanser and Slamanig~\cite{fuchsbauer2019structure} presented structure-preserving signatures on equivalence classes (SPS-EQ), which enable the joint randomization of messages and signatures without rendering them invalid. Combined with commitments, this approach forms the foundation for anonymous credential schemes.

\subsection{Issuer-Hiding Authentication} 

Recent research provided the issuer hiding property for stronger privacy protection for ACs. Schemes in~\cite{connolly2022improved,connolly2022protego,bosk2022hidden} achieved issuer hiding by raising the issuer's public key to some random power and proving the result is derived from a certain key in an accepted public key set. Bobolz et al.~\cite{bobolz2021issuer} used ZKP to prove that the holder owns a valid credential that can be verified by the public key signed by the verifier, which reduces the proving complexity to $O(1)$. Sanders et al. further optimized the proving and verification efficiency~\cite{sanders2023efficient} and designed a technique to verify the credential using the aggregated public key from approved issuers. Schemes in~\cite{mir2023aggregate,shi2023double} adopted the tag-based aggregatable mercurial signatures to achieve the aggregation of attributes from different issuers and issuer hiding. 

\subsection{Credential Revocation} 
At an early age, issuers maintain Certificate Revocation Lists (CRL), and the verifier uses the credential and Online Certificate Status Protocol (OCSP) to check the certificate status without considering identity privacy. Several subsequent research enhanced identity privacy by using domain-specific pseudonyms~\cite{kutylowski2011restricted}, trusted third parties~\cite{camenisch2003practical}, or Direct Anonymous Attestation (DAA)~\cite{brickell2004direct}, which only provide weak privacy protection or make strong security assumptions. For stronger privacy, some researchers require holders to provide non-revocation proofs in authentication based on whitelists~\cite{camenisch2002dynamic,fueyo2016efficiency,li2023attribute} or blacklists~\cite{nguyen2014u,jia2023generic,tsang2007blacklistable}. 

\section{Conclusion}
\label{sec: conclusion}

In this paper, we presented \Sysname, a novel issuer-hiding and revocable anonymous credential scheme for decentralized networks.
\Sysname introduces a flexible credential model that employs vector commitments with a padding strategy to unify credentials from heterogeneous issuers, enabling privacy-preserving authentication without enforcing a global static attribute set.
We also designed a decentralized revocation mechanism where holders prove non-revocation via sorted list gap proofs.
This approach effectively decouples revocation checks from verifier policies and prevents replay attacks while maintaining issuer anonymity.
Furthermore, \Sysname achieves strong attribute hiding through zk-SNARKs.
Our security analysis and performance evaluation confirm that \Sysname is a secure and efficient solution for privacy-preserving authentication in decentralized networks. 

\newpage
\bibliographystyle{IEEEtran}
\bibliography{ref}


\end{document}